\documentclass[12pt]{article}

\usepackage{amsmath, amssymb, mathtools, amssymb, amsthm}
\usepackage{natbib}

\newtheorem{theorem}{Theorem}

\begin{document}

\title{A simple framework for the axiomatization of
 exponential and quasi-hyperbolic discounting 
}

\author{Nina Anchugina\thanks{Department of Mathematics, The University of Auckland, Auckland, New Zealand; e-mail: {\tt n.anchugina@auckland.ac.nz}}}

\date{}

\maketitle

\begin{abstract}
The main goal of this paper is to investigate which normative requirements, or axioms, lead to exponential and quasi-hyperbolic forms of discounting. 
Exponential discounting has a well-established axiomatic foundation originally developed by \cite{koopmans1960stationary, koopmans1972representation} and \cite{koopmans1964stationary}   
with subsequent contributions by several other authors, including \citet{bleichrodt2008koopmans}. The papers by 
\citet{hayashi2003quasi} and 
\citet{olea2014axiomatization} axiomatize quasi-hyperbolic discounting.
The main contribution of this paper is to provide an alternative foundation for exponential and quasi-hyperbolic discounting, with simple, transparent axioms and relatively straightforward proofs. Using techniques by 
\citet{fishburn1982foundations} and 
\citet{harvey1986value}, we show that 
\citeauthor{anscombe1963definition}'s \citeyearpar{anscombe1963definition} version of Subjective Expected Utility theory can be readily adapted to axiomatize the aforementioned types of discounting, in both finite and infinite horizon settings.


\end{abstract}

\section{Introduction}
\label{intro}
The axiomatic foundation of intertemporal decisions is a fundamental question in economics and generates considerable research interest. Despite the fact that a number of possible ways of discounting have appeared in the literature so far, two types have been predominantly used: exponential discounting, first introduced by 
\citet{samuelson1937note}, and quasi-hyperbolic discounting \citep{phelps1968second, laibson1997golden}. The important question to be answered is which axioms allow us to say that the preferences of a decision maker can be represented using the discounted utility model with exponential or quasi-hyperbolic discount functions? Existing axiom systems for intertemporal decisions address this question. These systems can be roughly divided into two main groups: those with preferences over deterministic consumption streams and those with preferences over stochastic consumption streams.

The first group has been the leading approach in the area, both for exponential and quasi-hyperbolic functions. In this framework a consumption set is endowed with topological structure, and \citeauthor{debreu1960topological}'s \citeyearpar{debreu1960topological} theorem on additive representation is a key mathematical tool. 

Koopmans' result for exponential discounting with deterministic consumption streams \citep{koopmans1960stationary, koopmans1972representation, koopmans1964stationary}
remains the most well-known. A revised formulation of Koopmans' result was proposed by \citet{bleichrodt2008koopmans}, using alternative conditions on preferences. A similar approach was also suggested by \citet{harvey1986value}. The axiomatic foundation of exponential discounting for the special case of a single dated outcome was presented by \citet{fishburn1982time}. 

In a non-stochastic framework quasi-hyperbolic discounting, to our knowledge, has only been axiomatized by \citet{olea2014axiomatization}. Building on \citet{bleichrodt2008koopmans} they provide three alternative sets of axioms. \citeauthor{olea2014axiomatization}'s axiomatization will be discussed in more detail in Section \ref{sec:10}.

All the axiomatization systems mentioned above are formulated for infinite consumption streams. The finite horizon case has rarely been discussed. For exponential discounting, however, it can be found in \citet{fishburn1970utility}. 

The second group of axiomatic systems considers stochastic consumption \linebreak streams. To obtain an additive form the fundamental representation theorem of von Neumann and Morgenstern (vNM) \citeyearpar{morgenstern1947theory} is used.
 The application of this approach to exponential discounting was given by \citet{epstein1983stationary}. A consumption stream is considered to be an outcome of a lottery. The axiomatization of quasi-hyperbolic discounting by \citet{hayashi2003quasi} builds on \citeauthor{epstein1983stationary}'s \citeyearpar{epstein1983stationary} axiom system. Both Hayashi and Epstein axiomatize preferences over infinite stochastic consumption streams.
  
In this paper we work with preferences over streams of consumption lotteries, i.e., a setting in which there is a lottery in each period of time. In other words, we restrict Epstein and Hayashi's framework to product measures. This framework allows us to apply \citeauthor{anscombe1963definition}'s \citeyearpar{anscombe1963definition} result from Subjective Expected Utility Theory. The main advantage of this method is that it gives an opportunity to construct the discussed functional forms of discounting in a simpler way. Importantly, the present work establishes a unified treatment of exponential and quasi-hyperbolic discounting in both finite and infinite settings. With \citet{fishburn1982foundations} and \citet{harvey1986value} as the key sources of technical inspiration, our approach requires relatively simple axioms and facilitates proofs that are relatively straightforward. 
\section{Preliminaries}
\label{sec:1}
Assume that the objectives of a decision-maker can be expressed by a preference order $\succcurlyeq$ on the set of alternatives $X^n$, where $n$ may be $\infty$. Think of these alternatives as dated streams, for time periods $t \in \{1, 2, \ldots, n \}$\footnote{It should also be mentioned that our setting considers a discrete time space. A continuous-time framework can be found, for example, in the above-mentioned paper by \cite{fishburn1982time} and in a generalized model of hyperbolic discounting introduced by \cite{loewenstein1992anomalies}. \cite{harvey1986value} analyzes discrete sequences of timed outcomes with a continuous time space.
}.
We say that a utility function $U \colon X^n \to \mathbb{R}$ {\bf represents} this preference order, if for all ${\bf x}, {\bf y} \in X^n$, ${\bf x} \succcurlyeq {\bf y}$ if and only if $U({\bf x}) \geq U({\bf y})$.

We assume that $X$ is a {\bf mixture set}. That is, for every $x, y \in X$ and every $\lambda \in [0, 1]$, 
there exists $x \lambda y \in X$ satisfying: 
\begin{itemize} \renewcommand*\labelitemi{\textbullet}
\item $x1y=x$, 
\item $x \lambda y = y (1-\lambda) x$,
\item$(x \mu y) \lambda y = x (\lambda \mu) y$.
\end{itemize}
Since $X$ is a mixture set, the set $X^n$ is easily seen to be a mixture set under the following mixture operation: ${\bf x} \lambda {\bf y} = (x_1 \lambda y_1, \ldots, x_n \lambda y_n)$, where ${\bf x}, {\bf y} \in X^n$ and $\lambda \in [0, 1]$.

The utility function $u \colon X \to \mathbb{R}$ is called {\bf mixture linear} (or just linear, where no confusion is likely to arise) if for every $x, y \in X$ we have $u(x \lambda y)=\lambda u(x) + (1-\lambda) u(y)$ for every $\lambda \in [0, 1]$.

The binary relation $\succcurlyeq$ on $X^n$ induces a binary relation (also denoted $\succcurlyeq$) on $X$ in the usual way: for any $x, y \in X$ the preference $x \succcurlyeq y$ holds if and only if $(x, x,\ldots, x)\succcurlyeq (y, y,\ldots, y)$. 

The function $U$ is called a {\bf discounted utility function} if
\[
U({\bf x})=\sum_{t=1}^n D(t)u(x_t),
\]
for some non-constant $u \colon X \to \mathbb{R}$ and some $D\colon \mathbb{N}\to \mathbb{R}$ with $D(1)=1$. The function $D$ is called the {\bf discount function}. If $u$ is linear (and non-constant), then the function $U$ is called a {\bf discounted expected utility function}.

There are two types of discount functions which are commonly used in modeling of time preferences:

\begin{itemize} \renewcommand*\labelitemi{\textbullet}
\item {\bf Exponential discounting}: $D(t)=\delta^{t-1}$, where $\delta \in (0,1)$.
\item {\bf Quasi-hyperbolic discounting}: \begin{equation*} D(t) = 
	\left\{
					\begin{array}{lr}
						\ 1 &\ \text{if} \ t=1,\\			
						 \ \beta\delta^{t-1}&\text{if} \ t\geq 2.
					\end{array}
	\right.
\end{equation*}
for some $\delta \in (0, 1)$ and $\beta \in (0, 1]$.

\end{itemize}

The important characteristic of quasi-hyperbolic discounting is that it exhibits present bias. Present bias means that delaying two consumption streams from present ($t=1$) to the immediate future ($t=2$) can change the preferences of a decision-maker between these consumption streams.

The results of the recent experiments by \cite{chark2015extended} show that decision-makers are decreasingly impatient within the near future, however they discount the remote future at a constant rate.
In other words, present bias may extend over the present moment $(t=1)$ to the near future $t>2$, with a constant discount factor from some period $T$. This gives a further generalization of quasi-hyperbolic discounting, which we call {\bf semi-hyperbolic discounting}:
\begin{equation*} D(t) = 
	\left\{
					\begin{array}{lr}
						\ 1 &\ \text{if} \ t=1,\\			
						 \ \displaystyle\prod_{i=1}^{t-1}\beta_i\delta &\ \text{if} \ 1<t \leq T,\\
						 \ \delta^{t-T}\displaystyle\prod_{i=1}^{T-1}\beta_i\delta & \ \text{if } \ t>T.
					\end{array}
	\right.
\end{equation*}
We use $SH(T)$ to denote this discount function (for given $\delta, \beta_1, \ldots, \beta_{T-1}$). This form of discounting was previously applied to model the time preferences of a decision-maker in a consumption-savings problem \citep{young2007generalized}. Our $SH(T)$ specification is not quite the same as the notion of semi-hyperbolic discounting used in \cite{olea2014axiomatization}. They apply the term to any discount function which satisfies $D(t)=\delta^{t-T}D(T)$ for all $t>T$ (for some $T$). This class includes $SH(T)$, but is wider.
The possibility of generalizing quasi-hyperbolic discounting was earlier suggested by \cite{hayashi2003quasi}. The form of the discount function he proposed is:
\begin{equation*} D(t) = 
	\left\{
					\begin{array}{lr}
						\ 1 &\ \text{if} \ t=1,\\			
						 \ \displaystyle\prod_{i=1}^{t-1}\beta'_i &\ \text{if} \ 1<t \leq T,\\
						 \ \delta^{t-T}\displaystyle\prod_{i=1}^{T-1}\beta'_i & \ \text{if } \ t>T.
					\end{array}
	\right.
\end{equation*}
By substituting $\delta \beta_t = \beta'_t$ for all $t\leq {T-1}$ it is not difficult to see that semi-hyperbolic discounting $SH(T)$ coincides with the form suggested by \cite{hayashi2003quasi}. It is worth mentioning that he does not provide an axiomatization of this form of discounting, pointing out that this case is somewhat complicated. In our framework, however, the axiomatization of semi-hyperbolic discounting can be obtained as a relatively straightforward extension of the axiomatization of quasi-hyperbolic discounting.

The evidence of \citet{chark2015extended} on extended present bias suggests the following restrictions on the coefficients in $SH(T)$: $\beta_1<\beta_2<\ldots<\beta_{T-1}$. In our version of $SH(T)$ we will impose the weaker requirements $\beta_1 \leq \beta_2 \leq \ldots \leq \beta_{T-1}$, and $\beta_t\in (0, 1]$ for all $t\leq T-1$ and $\delta \in (0, 1)$.
Imposing these restrictions gives some advantages, as it can be immediately seen that exponential and quasi-hyperbolic discounting are the special cases of semi-hyperbolic discounting: $SH(1)$ is the exponential discount function, whereas $SH(2)$ is the quasi-hyperbolic discount function.

\section{AA representations}

We say that the preference order $\succcurlyeq$ on $X^n$ has an {\bf Anscombe and Aumann (AA) representation}, if for every ${\bf x}, {\bf y} \in X^n$:
\[ 
{\bf x}\succcurlyeq {\bf y} \text{ if and only if } \sum_{t=1}^ n w _tu(x_t) \geq \sum_{t=1}^ n w_tu(y_t), 
\]
where $u\colon X \to \mathbb{R}$ is non-constant and linear and  $w_t \geq 0$ for each $t$ with at least one $w_t>0$. We also say that the pair $(u, {\bf w})$ provides an AA representation for $\succcurlyeq$. 

A pre-condition for obtaining discounting in an exponential or quasi-hyperbolic form is additive separability. In the framework of preferences over streams of lotteries, \citeauthor{anscombe1963definition}'s \citeyearpar{anscombe1963definition} theorem provides axioms which give an additively separable representation when $n < \infty$. Anscombe and Aumann formulated their result for acts rather than temporal streams. Here, states of the world are replaced by time periods. 

\subsection{Finite case $(n< \infty)$}
\label{sec:2}
For $n<\infty$ the following axioms are necessary and sufficient for an AA representation:

\begin{description}
\item [Axiom F1.] (Weak order). $\succcurlyeq$ is a weak order on $X^n$.
\item [Axiom F2.] (Non-triviality). There exist some ${a, b} \in X$ such that 
\[
(a, a, \ldots, a) \succ (b, b, \ldots, b).
\] 
\item [Axiom F3.] (Mixture independence). ${\bf x}\succcurlyeq {\bf y}$ if and only if ${\bf x} \lambda {\bf z}\succcurlyeq {\bf y}  \lambda {\bf z}$ for every $\lambda \in (0, 1)$ and every ${\bf x}, {\bf y}, {\bf z} \in X^n$.
\item [Axiom F4.] (Mixture continuity). For every ${\bf x}, {\bf y}, {\bf z} \in X^n$ the sets \linebreak  $\{ \alpha \colon {\bf x} \alpha {\bf z} \succcurlyeq {\bf y} \}$ and $\{ \beta \colon  {\bf y} \succcurlyeq {\bf x} \beta {\bf z} \}$ are closed subsets of the unit interval.
\item [Axiom F5.] (Monotonicity). For every ${\bf x}, {\bf y} \in X^n$ if $x_t \succcurlyeq y_t$ for every $t$ then ${\bf x} \succcurlyeq {\bf y}$.
\end{description}

\begin{theorem}[AA]
\label{SEU}
The preferences $\succcurlyeq$ on $X^n$ satisfy axioms F1-F5 if and only if there exists an AA representation for $\succcurlyeq$ on $X^n$. 
If $(u, {\bf w})$ and $(u^{\prime}, {\bf w^{\prime}})$ both provide AA representations for $\succcurlyeq$ on $X^n$, then $u=Au^{\prime}+B$ for some $A>0$ and some $B$, and ${\bf w}=C {\bf w}^{\prime}$ for some $C>0$.
\end{theorem}
The proof of the theorem for the general mixture set environment can easily be constructed by combining the arguments in \citet{fishburn1982foundations} and \citet{ryan2009generalizations}.
 Evidently, the key axiom here is the condition of mixture independence. It is a strong axiom which imposes an additive structure.

\subsection{Infinite case ($n=\infty$)}
\label{sec:3}
Anscombe and Aumann's result may be extended to the infinite horizon case. One possible extension is given by \cite{fishburn1982foundations}. However, we give a slightly modified version which incorporates ideas from \cite{harvey1986value}. 

Fix some $x_0\in X$. We refer to the same $x_0$ throughout the rest of the paper.
A consumption stream ${\bf x}$ is called {\bf ultimately constant} if there exists 
$T$ such that ${\bf x}=(x_1, \ldots, x_T, x_0, x_0, \ldots)$. Note that there is a difference from the usage of this definition in \cite{bleichrodt2008koopmans} and \cite{olea2014axiomatization}, where $x_0$ can be arbitrary. Let $X_T$ be the set of ultimately constant consumption streams of length $T$. 
Denote the union of the sets $X_T$ over all $T$ as $X^{*}$. Let $X^{**}$ be the union of $X^*$ and all constant streams.
It is not hard to see that both $X^*, X^{**} \subset X^\infty$ are mixture sets.

We must mention that the fixed $x_0$ serves two purposes: firstly, it will be needed to state the convergence axiom; and secondly, it allows us to define the class $X^*$ of ultimately constant streams in a way that makes them a strict subset of the usually defined class. Since some of the axioms only restrict preferences over $X^{**}$ this second aspect confers some advantages.

\begin{description}
\item [Axiom I1.] (Weak order). $\succcurlyeq$ is a weak order on $X^\infty$.
\item [Axiom I2.] (Non-triviality). There exist some ${a, b} \in X$ such that \linebreak $a \succ x_0 \succ b$. 
\end{description}
Axiom I2 implies that $x_0$ is an interior point with respect to preference. It restricts both $\succcurlyeq$ and the choice of the fixed element $x_0$. 
\begin{description}
\item [Axiom I3.] (Mixture independence). ${\bf x}\succcurlyeq {\bf y}$ if and only if ${\bf x} \lambda {\bf z}\succcurlyeq {\bf y}  \lambda {\bf z}$ for every $\lambda \in (0, 1)$ and every ${\bf x},{\bf y}, {\bf z} 
\in X^{**}$.
\item [Axiom I4.] (Mixture continuity). For every ${\bf x}, {\bf z} \in X^{**}$ and every ${\bf y} \in X^\infty$  the sets $\{ \alpha \colon {\bf x} \alpha {\bf z} \succcurlyeq {\bf y} \}$ and $\{ \beta \colon  {\bf y} \succcurlyeq {\bf x} \beta {\bf z} \}$ 
are closed subsets of the unit interval.
\item [Axiom I5.] (Monotonicity). For every ${\bf x}, {\bf y} \in  X^\infty$: if $x_t \succcurlyeq y_t$ for every $t$ then ${\bf x} \succcurlyeq {\bf y}$.
\end{description}
We have applied a weaker version of the monotonicity axiom in comparison with the interperiod monotonicity used by Fishburn. However, Axiom I5 is sufficient to obtain an AA representation.

For the statement of the next axiom we need to introduce some notation. Let $[a]_k=(x_0, \ldots, x_0, a, x_0, \ldots)$ where $a\in X$ is in the k\textsuperscript{th} position. Using this notation, we state the following axiom: 

\begin{description}
\item [Axiom I6.](Convergence). For every ${\bf x} = (x_1, x_2, \ldots) \in X^\infty$, every $x^+, x^- \in X$ and every $k$:
\begin{itemize} \renewcommand*\labelitemi{\textbullet}
\item if $[x^+]_k\succ [x_k]_k$ there exists $T^+\geq k$ such that 
\[
{\bf x} \preccurlyeq {\bf x}_{k,T}^+ \ \text{for all }T \geq T^+,
\]
where ${\bf x}_{k,T}^+=(x_1, x_2, \ldots, x_{k-1}, x^+, x_{k+1}, \ldots, x_T, x_0, x_0, \ldots)$;
\item if $[x^-]_k \prec [x_k]_k$ there exists $T^- \geq k$ such that 
\[
{\bf x} \succcurlyeq {\bf x}_{k,T}^- \ \text{for all }T \geq T^-,
\]
where ${\bf x}_{k,T}^-=(x_1, x_2, \ldots, x_{k-1}, x^-, x_{k+1}, \ldots, x_T, x_0, x_0, \ldots)$.
\end{itemize}
\end{description}

Our convergence axiom differs from Axiom B6, that was used by Fishburn: 
\begin{description}
\item [Axiom B6.]
For some $\hat{x}\in X$, every ${\bf x}, {\bf y} \in X^\infty$ and every $\lambda \in (0, 1)$:
\begin{itemize} \renewcommand*\labelitemi{\textbullet}
\item if ${\bf x} \succ {\bf y}$, then there exists $T$ such that $(x_1, \ldots, x_n, \hat{x}, \hat{x}, \ldots) \succcurlyeq {\bf x} \lambda {\bf y}$ for all $n \geq T$;
\item if ${\bf x} \prec {\bf y}$, then there exists $T$ such that $(x_1, \ldots, x_n, \hat{x}, \hat{x}, \ldots) \preccurlyeq {\bf x} \lambda {\bf y}$ for all $n \geq T$.
\end{itemize} \end{description}
Instead, Axiom I6 adapts ideas from \cite{harvey1986value}\footnote{It is worth mentioning that Fishburn's motivation for the convergence axiom B6 looks somewhat contrived in the context of acts \citep[p. 113]{fishburn1982foundations}. However, it becomes very natural in the context where states of the world are re-interpreted as periods of time.}. Axiom I6 is more appealing for our purposes as it not only guarantees the convergence of the AA representation, but also allows us to relax two axioms, mixture independence and mixture continuity, which are no longer required to hold on all of $X^\infty$. 

We thus obtain the following representation:
\begin{theorem}[Infinite AA]
\label{ISEU}
The preferences $\succcurlyeq$ on $X^\infty$ satisfy axioms I1-I6 if and only if there exists an AA representation for $\succcurlyeq$ on $X^\infty$. 
If $(u, {\bf w})$ and $(u^{\prime}, {\bf w^{\prime}})$ both provide AA representations for $\succcurlyeq$ on $X^\infty$ , then $u=Au^{\prime}+B$ for some $A>0$ and some $B$, and $w=C w^{\prime}$ for some $C>0$.
\end{theorem}

The proof of Theorem \ref{ISEU} is given in the Appendix. It combines elements of the arguments in \citet{fishburn1982foundations}, \citet{harvey1986value} and \citet{ryan2009generalizations}.

\section{Discounted utility: finite case ($n<\infty$)}
\label{sec:4}
\subsection{Exponential discounting}
\label{sec:5}
Recall that a preference $\succcurlyeq$ on $X^n$ is represented by an exponentially discounted utility function if there exists a non-constant function $u \colon X \to \mathbb{R}$ and a parameter $\delta \in (0, 1)$ such that 
\[
U({\bf x}) = \sum _{t=1}^n \delta^{t-1}u(x_t). 
\]
If $u$ is linear (and non-constant), then we say that the pair $(u, {\bf \delta})$ provides an exponentially discounted expected utility representation.

Based on Theorem \ref{SEU} it is easy to obtain such a representation. In order to do so an adjustment of non-triviality and two additional axioms - impatience and stationarity - are required.
\begin{description}
\item [Axiom F2$^\prime$.] (Essentiality of period 1). There exist some $a, b \in X$ and some ${\bf x}\in X^n$ such that $(a, x_2, \ldots, x_n) \succ (b, x_2, \ldots, x_n)$. 
\end{description}

\begin{description}
\item [Axiom F6.] (Impatience). For all $a, b \in X$ if $a\succ b$, then for all ${\bf x} \in X^n$
\[
(a, b, x_3,\ldots,x_n) \succ (b, a, x_3, \ldots, x_n).
\] 
\item [Axiom F7.] (Stationarity). The preference ($a, x_2,\ldots, x_n)\succcurlyeq(a, y_2,\ldots, y_n)$ holds if and only if 
$(x_2,\ldots, x_n, a)\succcurlyeq(y_2,\ldots, y_n, a)$ for every $a\in X$ and every ${\bf x}, {\bf y} \in X^n$.
\end{description}

It is not hard to see that essentiality of each period $t$ follows from the essentiality of period 1 and the stationarity axiom.

Now the following result can be stated:

\begin{theorem}[Exponential discounting]
\label{EDF}
The preferences $\succcurlyeq$ on $X^n$ satisfy axioms F1, F2$^\prime$, F3-F7 if and only if there exists an exponentially discounted expected utility representation for $\succcurlyeq$ on $X^n$. If $(u, {\bf \delta})$ and $(u^{\prime}, {\bf \delta^{\prime}})$ both provide exponentially discounted expected utility representations for $\succcurlyeq$ on $X^n$, then $u=Au^{\prime}+B$ for some $A>0$ and some $B$, and $\delta=\delta'$.

\end{theorem}
\begin{proof}
It is straightforward to show that the axioms are implied by the representation. Conversely, suppose the axioms hold.
Note that non-triviality follows from essentiality of period 1 and monotonicity.

By Theorem 1 we therefore know that $\succcurlyeq$ has an AA representation $(u, {\bf w})$.
Define $\succcurlyeq^{\prime}$ on $X^{n-1}$ as follows:
\[
(x_1,\ldots, x_{n-1})\succcurlyeq^{\prime} (y_1,\ldots, y_{n-1}) \Leftrightarrow (x_0, x_1,\ldots, x_{n-1})\succcurlyeq(x_0, y_1,\ldots, y_{n-1}).
\]
Then $\succcurlyeq^{\prime}$ is represented by:
\[
U^{\prime} ({\bf x}) =w_2 u(x_1)+\ldots+w_n u(x_{n-1}). 
\]
Next, define $\succcurlyeq^{\prime \prime}$ on $X^{n-1}$ as follows:
\[
(x_1,\ldots, x_{n-1})\succcurlyeq^{\prime \prime} (y_1,\ldots, y_{n-1}) \Leftrightarrow (x_1,\ldots, x_{n-1}, x_0)\succcurlyeq(y_1,\ldots, y_{n-1}, x_0). 
\]
Then $\succcurlyeq^{\prime \prime}$ is represented by:
\[
U^{\prime \prime} ({\bf x})=w_1 u(x_1)+\ldots+w_{n-1} u(x_{n-1}).
\]
According to stationarity, these preferences are equivalent $\left( \succcurlyeq^{\prime}\equiv \succcurlyeq^{\prime \prime} \right )$ with two different AA representations ($U^{\prime}$ and $U^{\prime \prime}$). Preference orders $\succcurlyeq^{\prime}\equiv \succcurlyeq^{\prime\prime}$ satisfy the AA axioms on $X^{n-1}$.
Recall that $w_t$ are unique up to a scale. Hence, $w_{t+1}=\delta w_t$ for some $\delta>0$ and it follows that
\[
w_n=\delta w_{n-1}=\delta^2 w_{n-2}=\ldots= \delta^{n-t} w_t=\ldots=\delta^{n-1} w_1.
\]
Since all periods are essential it is without loss of generality to set $w_1=1$. Then we obtain the following representation for $\succcurlyeq$ on $X^n$:
\[
U({\bf x}) = \sum_{t=1}^n\delta^{t-1}u(x_t), \quad \text{where} \thickspace \delta>0.
\]
Since impatience holds: if $a\succ b$, then 
\[
(a, b, x_3,\ldots, x_n)\succ (b, a, x_3,\ldots, x_n).
\] 
From the representation it follows that:
\[
u(a)+\delta u(b) > u(b)+\delta u(a),
\]
 or, equivalently,
\[
(1-\delta) (u(a) - u(b)) > 0.
\]
As $u(a)>u(b)$ , it is possible to conclude that $\delta \in (0,1)$.

Suppose that $(u, {\bf \delta})$ and $(u^{\prime}, {\bf \delta^{\prime}})$ both provide exponentially discounted expected utility representations for $\succcurlyeq$ on $X^n$. Since $(u, {\bf \delta})$ and $(u^{\prime}, {\bf \delta^{\prime}})$ both provide AA representations for $\succcurlyeq$ it follows that  $u=Au^{\prime}+B$ for some $A>0$ and some $B$, and there is some $C>0$ such that $\delta^{t-1}=C(\delta^{\prime})^{t-1}$ for all $t$. Taking $t=1$ we obtain $C=1$, and hence $\delta=\delta^{\prime}$.
\end{proof}

\subsection{Semi-hyperbolic discounting}

A preference $\succcurlyeq$ on $X^n$ has a $SH(T)$ discounted utility representation if there exists a non-constant function $u \colon X \to \mathbb{R}$ and parameters $\beta_1 \leq \beta_2 \leq \ldots \leq \beta_{T-1}$, and $\beta_t\in (0, 1]$ for all $t\leq T-1$ and $\delta \in (0, 1)$ such that the following function represents $\succcurlyeq$:
\begin{equation*}
\begin{split}
U({\bf x}) = u(x_1)+\beta_1 \delta u(x_2)+\beta_1 \beta_2 \delta^2 u(x_3)+\ldots+\beta_1 \beta_2\cdots \beta_{T-2} \delta^{T-2}u(x_{T-1}) \\ +\beta_1 \beta_2\cdots \beta_{T-1} \sum _{t=T}^n \delta^{t-1}u(x_t).
\end{split}
\end{equation*}
If $u$ is linear (and non-constant), then the function $U$ is called a {\bf $SH(T)$ discounted expected utility representation}. In this case, we say that $(u, \boldsymbol{\beta}, \delta)$ provides a $SH(T)$ discounted expected utility representation, where $\boldsymbol{\beta}=(\beta_1, \beta_2, \ldots, \beta_{T-1})$.

To obtain this form of discounting a number of modifications to the set of axioms is required. A stronger essentiality condition should be used:
\begin{description}
\item [Axiom F2$^{\prime\prime}$.] (Essentiality of periods $1,\ldots, T$). 
There exist some $a, b \in X$ and some ${\bf x}\in X^n$ such that for every $t=1,\ldots, T$:
\[
(x_1, x_2, \ldots, x_{t-1}, a, x_{t+1}, \ldots, x_n) \succ (x_1, x_2, \ldots, x_{t-1}, b, x_{t+1}, \ldots, x_n).
\]
\end{description}

The impatience axiom, which is used to guarantee $\delta \in (0, 1)$, should be restated for the periods $T$ and $T+1$:

\begin{description}
\item [Axiom F6$^{\prime}$.] (Impatience). For every $a, b \in X$ if $a\succ b$, then for every ${\bf x} \in X^n$:
\[
(x_1,\ldots, x_{T-1}, a, b, x_{T+2},\ldots,x_n)  \succ (x_1, \ldots, x_{T-1}, b, a, x_{T+2},\ldots,x_n).
\] 
\end{description}
The generalization requires relaxing the axiom of stationarity to stationarity from period $T$. 
\begin{description}
\item [Axiom F7$^{\prime}$.] (Stationarity from period T).
The preference
\[
(x_1, \ldots, x_{T-1}, a, x_{T+1},\ldots, x_n)\succcurlyeq (x_1, \ldots, x_{T-1}, a, y_{T+1},\ldots, y_n)
\]
holds if and only if 
\[
(x_1, \ldots, x_{T-1}, x_{T+1},\ldots, x_n, a)\succcurlyeq (x_1, \ldots, x_{T-1}, y_{T+1},\ldots, y_n, a)
\]
for every $a \in X$ and every ${\bf x} \in X^n$.
\end{description}
The addition of the early bias axiom is needed, so that present bias may arise between any periods $\{t, t+1\}$, where $t\leq T$.
\begin{description}
\item [Axiom F8.] (Early bias) For every $a, b, c, d \in X$ such that $a\succ c, b\prec d$, for all ${\bf x} \in X^n$ and every $t \leq T$ if 
\[
(x_1, \ldots, x_{t-1}, a, b, x_{t+2}, \ldots, x_n)\sim (x_1, \ldots, x_{t-1}, c, d, x_{t+2}, \ldots, x_n), \ \text{then}
\]
\end{description}
\[
(x_1, \ldots, x_{t-2}, a, b, x_{t+2}, \ldots, x_n, x_{t-1}) \succcurlyeq (x_1, \ldots, x_{t-2}, c, d, x_{t+2}, \ldots, x_n, x_{t-1}).
\]

The early bias axiom is also referred to as the extended present bias axiom.

\begin{theorem}[Semi-hyperbolic discounting]
\label{SHDF}
The preferences $\succcurlyeq$ on $X^n$ satisfy axioms F1, F2$^{\prime\prime}$, F3, F4, F5, F6$^{\prime}$, F7$^{\prime}$, F8 if and only if there exists a $SH(T)$ discounted expected utility representation for $\succcurlyeq$ on $X^n$. If $(u, \boldsymbol{\beta}, \delta)$ and $(u^{\prime}, \boldsymbol{\beta'}, \delta^{\prime})$ both provide $SH(T)$ discounted expected utility representations for $\succcurlyeq$ on $X^n$, then $u=Au^{\prime}+B$ for some $A>0$ and some $B$, and $\delta=\delta'$, $\boldsymbol{\beta}=\boldsymbol{\beta'}$.
\end{theorem}
\begin{proof}
It can be easily seen that the axioms are implied by the representation. Suppose that the axioms hold.
As for Theorem \ref{EDF}, the conditions of AA representation are satisfied, so it follows that $\succcurlyeq$ has an AA representation $({\bf w}, u)$. Define $\succcurlyeq'$ on $X^{n-T}$ as follows:
\[
(x_{1},\ldots, x_{n-T})\succcurlyeq' (y_{1},\ldots, y_{n-T}) \Leftrightarrow 
\]
\[
(x_0, \ldots, x_0, x_{1},\ldots, x_{n-T})\succcurlyeq (x_0, \ldots, x_0, y_{1},\ldots, y_{n-T}). 
\]
Then $\succcurlyeq'$ is represented by:
\[
U'({\bf x})=w_{T+1} u(x_{1})+\ldots+w_n u(x_{n-T}). 
\]
Next, define $\succcurlyeq''$ on $X^{n-T}$ as follows:
\[
(x_{1},\ldots, x_{n-T})\succcurlyeq'' (y_{1},\ldots, y_{n-T}) \Leftrightarrow 
\]
\[
(x_0, \ldots, x_0, x_{1},\ldots, x_{n-T}, x_0)\succcurlyeq (x_0, \ldots, x_0, y_{1},\ldots, y_{n-T}, x_0).
\]
Then $\succcurlyeq''$ is represented by:
\[
U''({\bf x})=w_T u(x_{1})+\ldots+w_{n-1} u(x_{n-T}).
\]
According to stationarity from period $T$, the preferences are equivalent \linebreak $\left(\succcurlyeq'\equiv \succcurlyeq''\right)$ with two different AA representations ($U'$ and $U''$).

Preference orders $\succcurlyeq'\equiv \succcurlyeq''$ satisfy the AA axioms on $X^{n-T}$. Recall that $w_t$ are unique up to a scale. Hence, as essentiality holds for all $t$ (which follows from Axiom F2$^{\prime}$ and Axiom F7$^{\prime}$), we have $w_{t+1}=\delta w_t$ for some $\delta>0$ and hence
\[
w_n=\delta w_{n-1}=\delta^2 w_{n-2}=\ldots= \delta^{n-t} w_t=\ldots=\delta^{n-T} w_T.
\]
Therefore, $w_t=\delta^{t-T}w_T$ for all $t\geq T+1$. We therefore obtain the following representation for $\succcurlyeq$:
\begin{equation*}
U({\bf x}) = w_1 u(x_1)+\ldots + w_{T-1} u(x_{T-1}) + w_T \sum_{t=T}^n \delta^{t-T} u(x_t).
\end{equation*}
Because of the essentiality of the first period and uniqueness of $u$ up to affine transformations:
\[
\hat U({\bf x}) = u(x_1)+\frac{w_2}{w_1}u(x_2)+\ldots+\frac{w_{T-1}}{w_1}u(x_{T-1})+ \frac{w_T}{w_1} \sum_{t=T}^n \delta^{t-T} u(x_t). 
\]
Note that 
\begin{align*}
&\frac{w_3}{w_1}=\frac{w_3}{w_2} \cdot \frac{w_2}{w_1},\\
&\cdots, \\
& \frac{w_T}{w_1}=\frac{w_T}{w_{T-1}}\cdot \frac{w_{T-1}}{w_{T-2}} \cdot \ldots \cdot \frac{w_2}{w_1}. 
\end{align*}

Let $\gamma_{t-1}=\frac{w_t}{w_{t-1}}$ for all $t\leq T$. Therefore,
\begin{align*}
&\frac{w_2}{w_1}=\gamma_1, \\
&\frac{w_3}{w_1}=\gamma_1 \gamma_2, \\
&\cdots, \\
& \frac{w_T}{w_1}=\gamma_1 \gamma_2 \ldots \gamma_{T-1}. 
\end{align*}
With this notation: 
\[
\hat U({\bf x}) = u(x_1)+\gamma_1 u(x_2)+\ldots + \gamma_1\cdots\gamma_{T-2} u(x_{T-1}) + \gamma_1\cdots\gamma_{T-1}\sum_{t = T}^n \delta^{t-T} u(x_t).
\]
It is necessary to show that $\gamma_{t-1}=\beta_{t-1}\delta$  with $\beta_{t-1}\in (0,1]$ for all $t\leq T$. 

Suppose that $t=T$. Choose $a, b, c, d \in X$ such that $u(b)<u(d)$, \linebreak $u(a)>u(c)$ and
\begin{equation}\label{eq:pb}
\gamma_1\cdots\gamma_{T-1} u(a)+\gamma_1\cdots\gamma_{T-1}\delta u(b)=\gamma_1\cdots\gamma_{T-1} u(c)+\gamma_1\cdots\gamma_{T-1}\delta u(d).
\end{equation}
Since essentiality is satisfied for each period we can rearrange the equation \eqref{eq:pb}:
\begin{equation}
\label{eq:pb1}
\delta= \frac{u(a)-u(c)}{u(d)-u(b)}.
\end{equation}
From \eqref{eq:pb} it also follows that 
\[
(x_1, \ldots, x_{T-1}, a, b, x_{T+2}, \ldots, x_n)\sim (x_1, \ldots, x_{T-1}, c, d, x_{T+2}, \ldots, x_n), 
\]
Therefore, by the early bias axiom:
\[
(x_1, \ldots, x_{T-2}, a, b, x_{T+2}, \ldots, x_n, x_{T-1}) \succcurlyeq (x_1, \ldots, x_{T-2}, c, d, x_{T+2}, \ldots, x_n, x_{T-1}).
\]
Thus we obtain:
\[
\gamma_1\cdots\gamma_{T-2} u(a)+\gamma_1\cdots\gamma_{T-1} u(b) \geq \gamma_1\cdots\gamma_{T-2} u(c)+\gamma_1\cdots\gamma_{T-1} u(d).
\]
Since the essentiality condition is satisfied for each period we can rearrange this inequality:
\begin{equation}
\label{eq:pb2}
\gamma_{T-1} \leq  \frac{u(a)-u(c)}{u(d)-u(b)}.
\end{equation}
Comparing \eqref{eq:pb1} to \eqref{eq:pb2} we conclude that $\delta \geq \gamma_{T-1}$, therefore, $\gamma_{T-1}=\beta_{T-1} \delta$, where $\beta_{T-1}\in (0, 1]$. 

Analogously, suppose that $t=T-1$. Choose $a', b', c', d' \in X$ such that $u(b')<u(d')$ and $u(a')>u(c')$. Using present bias axiom and essentiality of each period we obtain  
\begin{equation}
\label{eq:pb3}
\gamma_{T-1}=\frac{u(a')-u(c')}{u(d')-u(b')},
\end{equation}
and
\begin{equation}
\label{eq:pb4}
\gamma_{T-2} \leq \frac{u(a')-u(c')}{u(d')-u(b')}.
\end{equation}
It follows from  \eqref{eq:pb3} and  \eqref{eq:pb4} that $\gamma_{T-2} \leq \gamma_{T-1}$. Therefore, $\gamma_{T-2} = \beta'_{T-2} \gamma_{T-1}$, where $\beta'_{T-2}\in (0, 1]$. Recall that $\gamma_{T-1}=\beta_{T-1} \delta$. Hence, 
\[
\gamma_{T-2} = \beta'_{T-2} \beta_{T-1} \delta =\beta_{T-2} \delta, 
\]
where $\beta_{T-2} = \beta'_{T-2} \beta_{T-1}$ and $\beta_{T-2} \in (0,1]$ as both $\beta'_{T-2}\in (0, 1]$ and $\beta_{T-1}\in (0, 1]$. Note also that $\beta_{T-2}\leq \beta_{T-1}$.

Using the early bias axiom repeatedly for $t<T-1$ we obtain $\gamma_{t-1}=\beta_{t-1}\delta$  with $\beta_{t-1}\in (0,1]$ for all $t\leq T$ and $\beta_1 \leq \beta_2 \leq \ldots \leq \beta_{T-1}$. Hence, 
\begin{align*}
\hat U({\bf x}) = u(x_1)+\beta_1 \delta u(x_2)+\beta_1 \beta_2 \delta^2 u(x_3)+\ldots+\beta_1 \beta_2\cdots \beta_{T-2} \delta^{T-2}u(x_{T-1}) \\+\beta_1 \beta_2\cdots \beta_{T-1}\sum _{t=T}^n \delta^{t-1}u(x_t).
\end{align*}

To show that $\delta \in (0,1)$ the impatience axiom should be applied. For every $a, b \in X$ if $a\succ b$, then for every ${\bf x} \in X^n$
\[
(x_1,\ldots, x_{T-1}, a, b, x_{T+2},\ldots,x_n)  \succ (x_1, \ldots, x_{T-1}, b, a, x_{T+2},\ldots,x_n).
\] 
Then
\[
\beta_1\cdots \beta_{T-1} \delta^{T-1} u(a) + \beta_1\cdots \beta_{T-1} \delta^T u(b) > \beta_1\cdots \beta_{T-1} \delta^{T-1} u(b) + \beta_1\cdots \beta_{T-1} \delta^T u(a).
\]
Therefore, due to essentiality of each period:
\[
(1-\delta)(u(a) - u(b)) > 0.
\]
Hence, $\delta \in (0, 1)$.

Suppose that $(u, \boldsymbol{\beta}, \delta)$ and $(u^{\prime}, \boldsymbol{\beta'}, \delta^{\prime})$ both provide $SH(T)$ discounted expected utility representations for $\succcurlyeq$ on $X^n$.  Let $D(t)$ and $D^{\prime}(t)$ be semi-hyperbolic discount functions for given $\boldsymbol{\beta}, \delta$ and $\boldsymbol{\beta^{\prime}}, \delta^{\prime}$, respectively. Since $(u, \boldsymbol{\beta}, \delta)$ and $(u^{\prime}, \boldsymbol{\beta'}, \delta^{\prime})$ both provide AA representations for $\succcurlyeq$ it follows that  $u=Au^{\prime}+B$ for some $A>0$ and some $B$, and there is some $C>0$ such that $D(t)=C\cdot D^{\prime}(t)$ for all $t$. Taking $t=1$ we obtain $C=1$, and hence, letting $t=2, 3, \ldots, T$ we get $\beta_t \delta = \beta^{\prime}_t \delta^{\prime}$ for all $t\leq T$.
Finally, letting $t=T+1$ we conclude that $\delta = \delta^{\prime}$. Therefore, $\boldsymbol{\beta}=\boldsymbol{\beta^{\prime}}$.
\end{proof}

\section{Discounted utility: infinite case ($n=\infty$)}
\label{sec:7}
\subsection{Exponential discounting}
\label{sec:8}
Based on the AA representation for the preferences over infinite consumption streams (Theorem \ref{ISEU}), with some strengthening of non-triviality (Axiom I2) and the addition of a suitable stationarity axiom, discounting functions in an exponential form can be obtained. The impatience axiom is not needed since convergence (Axiom I6) plays its role.

\begin{description}
\item [Axiom I2$^\prime$.] (Essentiality of period 1). There exist some $a, b \in X$ such that $[a]_1 \succ x_0 \succ [b]_1$. 
\end{description}

\begin{description}
\item [Axiom I7.] (Stationarity). The preference $(a, x_1, x_2, \ldots) \succcurlyeq (a, y_1, y_2, \ldots)$ holds if and only if $ (x_1, x_2, \ldots)  \succcurlyeq (y_1, y_2, \ldots)$ for every $a \in X$ and every ${\bf x}, {\bf y} \in X^\infty$.
\end{description}
\begin{theorem}[Exponential discounting]
\label{EDI}
The preferences $\succcurlyeq$ on $X^\infty$ satisfy axioms I1, I2$^\prime$, I3-I7 if and only if there exists an exponentially discounted expected utility representation for $\succcurlyeq$ on $X^\infty$. If $(u, {\bf \delta})$ and $(u^{\prime}, {\bf \delta^{\prime}})$ both provide exponentially discounted expected utility representations for $\succcurlyeq$ on $X^\infty$, then $u=Au^{\prime}+B$ for some $A>0$ and some $B$, and $\delta=\delta^{\prime}$.
\end{theorem}
\begin{proof}
The necessity of the axioms is straightforward.
The proof of sufficiency follows the steps of the proof of Theorem \ref{EDF} with $n=\infty$.
Applying Theorem \ref{ISEU} to the preferences satisfying the stationarity axiom we obtain the representation:
\[
U({\bf x}) = \sum_{t=1}^\infty\delta^{t-1}u(x_t), 
\]
where $\delta>0$ and ${\bf x} \in X^\infty$.

Next, instead of using the impatience axiom as it is done in the finite case, the convergence axiom is applied. 
Take a constant stream ${\bf a}=(a, a, \ldots)$, such that $u(a) \neq 0$. Then,
\[
U({\bf a}) = \sum_{t=1}^\infty\delta^{t-1}u(a) = u(a) \sum_{t=1}^\infty\delta^{t-1}, 
\]
We know that $U({\bf a})$ should converge by Theorem \ref{ISEU}. It follows that $\delta <1$.

The proof of the uniqueness claims is analogous to Theorem \ref{EDF}.
\end{proof}

\subsection{Semi-hyperbolic discounting}

The extension of semi-hyperbolic discounting to the case where $n=\infty$ is easily obtained.
\begin{description}
\item [Axiom I2$^{\prime\prime}$.] (Essentiality of periods $1,\ldots, T$). For some $a, b \in X$ we have $[a]_t\succ x_0 \succ [b]_t$ for every $t=1,\ldots, T$.
\end{description}
The generalization requires relaxing the axiom of stationarity to stationarity from period $T$. 
\begin{description}
\item [Axiom I7$^{\prime}$.] (Stationarity from period T).
The preference
\[
(x_1, \ldots, x_{T-1}, a, x_{T+1},\ldots)\succcurlyeq (x_1, \ldots, x_{T-1}, a, y_{T+1},\ldots)
\] 
holds if and only if 
\[
(x_1, \ldots, x_{T-1}, x_{T+1},\ldots)\succcurlyeq (x_1, \ldots, x_{T-1}, y_{T+1},\ldots)
\]
for every $a \in X$, and every ${\bf x} \in X^{\infty}$.
\end{description}
As in the finite case the addition of the early bias axiom allows present bias between $\{t, t+1\}$, where $t\leq T$.
\begin{description}
\item [Axiom I8.] (Early bias) For every $a, b, c, d \in X$ such that $a\succ c, b\prec d$, and for all ${\bf x} \in X^{\infty}$ and every $t \leq T$ 
\[
\ \text{if} \ (x_1, \ldots, x_{t-1}, a, b, x_{t+2}, \ldots)\sim (x_1, \ldots, x_{t-1}, c, d, x_{t+2}, \ldots), \ \text{then}
\]
\[
(x_1, \ldots, x_{t-2}, a, b, x_{t+2}, \ldots) \succcurlyeq (x_1, \ldots, x_{t-2}, c, d, x_{t+2}, \ldots).
\]
\end{description}

\begin{theorem}[Semi-hyperbolic discounting]
The preferences $\succcurlyeq$ on $X^{\infty}$ satisfy axioms I1, I2$^{\prime\prime}$, I3-I6, I7$^{\prime}$, I8 if and only if there exists a
$SH(T)$ discounted expected utility representation for $\succcurlyeq$ on $X^n$. If $(u, \boldsymbol{\beta}, \delta)$ and $(u^{\prime}, \boldsymbol{\beta'}, \delta^{\prime})$ both provide $SH(T)$ discounted expected utility representations for $\succcurlyeq$ on $X^n$, then $u=Au^{\prime}+B$ for some $A>0$ and some $B$, and $\delta=\delta'$, $\boldsymbol{\beta}=\boldsymbol{\beta'}$.
\end{theorem}
\begin{proof} The necessity of the axioms is obviously implied by the representation.
The proof of sufficiency is analogous to the finite case. Applying Theorem \ref{ISEU} and stationarity from period $T$ we get the representation:
\begin{equation*}
U({\bf x}) = w_1 u(x_1)+\ldots + w_{T-1} u(x_{T-1}) + w_T \sum_{t=T}^{\infty} \delta^{t-T} u(x_t).
\end{equation*}

Next, dividing by $w_1>0$ and introducing the notation $\frac{w_{t}}{w_{t-1}}=\gamma_{t-1}>0$, where $t\leq T$, the representation becomes
\[
\hat U({\bf x}) = u(x_1)+\gamma_1 u(x_2)+\ldots + \gamma_1\cdots\gamma_{T-2} u(x_{T-1}) + \gamma_1\cdots\gamma_{T-1} \sum_{t = T}^{\infty} \delta^{t-T} u(x_t).
\]

Using essentilaity of each period and the early bias axiom repeatedly, we demonstrate that $\gamma_{t-1}=\beta_{t-1}\delta$ with $\beta_{t-1}\in (0,1]$ for all $t\leq T$ and $\beta_1 \leq \beta_2 \leq \ldots \leq \beta_{T-1}$.
Therefore,
\begin{align*}
\hat U({\bf x}) = u(x_1)+\beta_1 \delta u(x_2)+\beta_1 \beta_2 \delta^2 u(x_3)+\ldots+\beta_1 \beta_2\cdots \beta_{T-2} \delta^{T-2}u(x_{T-1}) \\ +\beta_1 \beta_2\cdots \beta_{T-1}\sum _{t=T}^{\infty} \delta^{t-1}u(x_t).
\end{align*}

Finally, to show that $\delta \in (0, 1)$, take a constant stream ${\bf a}=(a, a, \ldots)$, such that $u(a)\neq 0$.
Then,
\begin{align*}
\hat{U}({\bf a}) =& u(a)+\beta_1 \delta u(a)+\ldots+\beta_1\cdots\beta_{T-2}\delta^{T-2} u(a) + \beta_1\cdots\beta_{T-1} \sum_{t = T}^{\infty} \delta^{t-1} u(a)\\
=& u(a) \left (1+\beta_1 \delta+\ldots+\beta_1\cdots\beta_{T-2}\delta^{T-2} + \beta_1\cdots\beta_{T-1} \sum_{t = T}^{\infty} \delta^{t-1} \right).
\end{align*}
From Theorem 2 it follows that $\hat{U}({\bf a})$ converges, therefore, $\delta<1$.

The proof of the uniqueness claims is analogous to Theorem \ref{SHDF}.
\end{proof}

\section{Discussion}
\label{sec:10}
A number of axiomatizations of exponential and quasi-hyperbolic discounting have been suggested by different authors. 
In fact, all the axiomatizations use different assumptions and there is no straightforward transformation from one type of discounting to another. In this paper we provided an alternative approach to get a time separable discounted utility representation, showing that Anscombe and Aumann's result can be exploited as a common background for axiomatizing exponential and quasi-hyperbolic discounting in both finite and infinite time horizons. In addition, we demonstrated that the axiomatization of quasi-hyperbolic discounting can be easily extended to $SH(T)$.

A key distinguishing feature of our set-up is the mixture set structure for $X$ and the use of the mixture independence condition. An essential question, however, is whether mixture independence is normatively compelling in a time preference context, because states are mutually exclusive whereas time periods are not. It is worth mentioning that the temporal interpretation of the AA framework was also used by \citet{wakai2008model} to axiomatize an entirely different class of preferences, which exhibit a desire to spread bad and good outcomes evenly over time.

Commonly, the condition of joint independence is used to establish additive separability in time-preference models.
Given $A \subseteq T$, where $T=\{1, \ldots, n\}$, and ${\bf x}, {\bf y} \in X^n$, define ${\bf x}_{A}{\bf y}$ as follows: $(x_{A}y)_t$ is $x_t$ if $t\in A$ and $y_t$ otherwise.
The preference order $\succcurlyeq$ satisfies {\bf joint independence} if for every $A\subseteq T$ and for every ${\bf x}, {\bf x'}, {\bf y}, {\bf y'} \in X^n$: 
\begin{eqnarray*}
\label{eq:JI}
{\bf x}_A{\bf y} \succcurlyeq {\bf x'}_A{\bf y} \text{ if and only if } {\bf x}_A{\bf y'} \succcurlyeq {\bf x'}_A{\bf y'}.
\end{eqnarray*}

Joint independence is used to obtain an additively separable representation by \citet{debreu1960topological}, so we will sometimes refer it as a Debreu-type independence condition. 
It is known that mixture independence implies joint independence \citep{grant2009zandt}, but whether joint independence (with some other plausible conditions) implies mixture independence is yet to be determined.   

In fact, we are not the first to use a mixture-type independence condition in the context of time preferences. \citet{wakai2008model} also does so, though he uses the weaker form of \emph{constant independence} introduced by \citet{gilboa1989maxmin}. 

A version of the mixture independence condition can also be formulated in a Savage environment \citep{savage1954foundations} without objective probabilities, as discussed in \citet{gul1992savages}. \citet{olea2014axiomatization} use precisely this version of mixture independence in one of their axiomatizations of quasi-hyperbolic discounting. For every $x, y \in X$ let us write $(x, y)$ for $(x, y, y, \ldots)\in X^\infty$. Let $m(x_1,y_1)$ denote some $c \in X$ satisfying $(x_1, y_1) \sim (c, c)$. For any streams $(x_1, x_2)$ and $(z_1, z_2)$ the consumption stream $(m(x_1,z_1), m(x_2,z_2))$ is called a subjective mixture of $(x_1, x_2)$ and $(z_1, z_2)$. Olea and Strzalecki's version of the mixture independence axiom (their Axiom I2) is as follows:
for every $x_1, x_2, y_1, y_2, z_1, z_2 \in X$ if 
$(x_1, x_2) \succcurlyeq (y_1, y_2)$, then 
\[
(m(x_1,z_1), m(x_2,z_2)) \succcurlyeq (m(y_1,z_1), m(y_2,z_2))  
\]
and
\[
(m(z_1,x_1), m(z_2,x_2)) \succcurlyeq (m(z_1,y_1), m(z_2,y_2)).
\]
In other words, if a consumption stream $(x_1, x_2)$ is preferred to a stream $(y_1, y_2)$, then subjectively mixing each stream with $(z_1, z_2)$ does not affect the preference.

In their axiomatization of quasi-hyperbolic discounting Olea and Strzalecki invoke their mixture independence condition (Axiom 12) as well as Debreu-type independence conditions. The latter are used to obtain a representation in the form
\[
{\bf x}\succcurlyeq {\bf y} \text{ if and only if }  u(x_1)+\sum_{t=2}^{\infty} \delta^{t-1}v(x_t) \geq u(y_1)+\sum_{t=2}^{\infty} \delta^{t-1}v(y_t),
\]
\and then their Axiom 12 is used to ensure $v=\beta u$.\footnote{As pointed out above, mixture independence stated for $n$ periods implies joint independence for $n$ periods. Hence, this raises the obvious question of whether it is possible to use an n-period version of the subjective mixture independence axiom to obtain a time separable discounted utility representation without the need for the Debreu-type independence conditions.} 

\citet{hayashi2003quasi} and \citet{epstein1983stationary} considered preferences over lotteries over consumption streams. 
In their framework $X^\infty$ is the set of non-stochastic consumption streams, where $X$ is required to be a compact connected separable metric space. Denote the set of probability measures on Borel $\sigma$-algebra defined on $X^\infty$ as $\Delta(X^\infty)$. It is useful to note that our setting is the restriction of the \citeauthor{hayashi2003quasi} and \citeauthor{epstein1983stationary} set-up to product measures, i.e., to $\Delta(X)^\infty \subset \Delta(X^\infty)$.
The axiomatization systems by \citeauthor{hayashi2003quasi} and \citeauthor{epstein1983stationary} are based on the assumptions of expected utility theory. The existence of a continuous and bounded vNM utility index $U \colon \Delta(X^\infty) \to \mathbb{R}$ is stated as one of the axioms. A set of necessary and sufficient conditions for this is provided by \citet{grandmont1972continuity}, and includes the usual vNM independence condition on $\Delta(X^\infty)$:
for every $ {\bf x}, {\bf y}, {\bf z} \in \Delta(X^\infty)$ and any $\alpha \in [0, 1]$,
${\bf x} \sim {\bf y}$ implies $\alpha {\bf x} + (1-\alpha) {\bf z} \sim \alpha {\bf y} + (1-\alpha) {\bf z}$.

Obviously, this independence condition is not strong enough to deliver joint independence of time periods, which is why additional assumptions of separability are needed. Two further Debreu-type independence conditions are required for exponential discounting:
\begin{itemize} \renewcommand*\labelitemi{\textbullet}
\item independence of stochastic outcomes in periods $\{1, 2\}$ from deterministic outcomes in $\{3, 4, \ldots \}$,
\item independence of stochastic outcomes in periods $\{2, 3, \ldots\}$ from deterministic outcomes in period $\{1\}$.
\end{itemize}

To obtain quasi-hyperbolic discounting two additional Debreu-type independence conditions should be satisfied:
\begin{itemize} \renewcommand*\labelitemi{\textbullet}
\item independence of stochastic outcomes in periods $\{2, 3\}$ from deterministic outcomes in periods $\{1\}$ and $\{4, \ldots \}$,  
\item independence of stochastic outcomes in periods $\{3, 4, \ldots\}$ from deterministic outcomes in periods $\{1, 2\}$.
\end{itemize}

It is easy to see that these axioms applied to the non-stochastic consumption streams are analogous to the Debreu-type independence conditions used in \citet{bleichrodt2008koopmans} and \citet{olea2014axiomatization}.

In summary, to get a discounted utility representation with the discount function in either exponential and quasi-hyperbolic form separability must be assumed. The mixture independence axiom appears to be a strong assumption, however, it gives the desired separability without the need for additional Debreu-type independence conditions.

\section{Appendix: Proof of Theorem \ref{ISEU}}
\label{sec:11}
\begin{proof}
Necessity of the axioms is straightforward to verify. Therefore we will focus on the proof of sufficiency.
\paragraph{Step 1.}  Applying Theorem 1 of \citet{fishburn1982foundations} to the mixture set $X$, it follows from Axioms I1, I3, I4 that there exists a linear utility function $u$ preserving the order on $X$ (unique up to positive affine transformations). Normalize $u$ so that $u(x_0)=0$. Note that by non-triviality $u(x_0)$ is in the interior of the non-degenerate interval $u(X)$.

Convert streams into their utility vectors by replacing the outcomes in each period by their utility values. Define the following order:
$(v_1, v_2, \ldots) \succcurlyeq^* (u_1, u_2, \ldots) \Leftrightarrow$ there exist ${\bf x, y} \in X^{\infty}$ such that ${\bf x} \succcurlyeq {\bf y}$ and $u(x_t)=v_t$ and $u(y_t)=u_t$ for every $t$. This order is unambiguously defined because of the monotonicity assumption, i.e., if $x_i \sim x^{\prime}_i$ then $(x_1, \ldots, x_i, \ldots) \sim (x_1, \ldots, x^{\prime}_i, \ldots)$.

The preference order $\succcurlyeq^*$ inherits the properties of weak order, mixture independence and mixture continuity from $\succcurlyeq$. Note that $u(X)^{\infty}$ is a mixture set under the standard operation of taking convex combinations:
if ${\bf v, u} \in u(X)^{\infty}$ then 
\[
{\bf v}\lambda {\bf u}=\lambda {\bf v} + (1-\lambda) {\bf u} \ \text{for every} \ \lambda \in (0,1). 
\]
Therefore, by Theorem 1 of \citet{fishburn1982foundations} we obtain a linear representation $U \colon$ $u(X)^{\infty} \to \mathbb{R}$, where $U$ is unique up to positive affine transformations. 

Hence ${\bf v} \succcurlyeq^* {\bf u}$ if and only if $U({\bf v}) \geq U({\bf u})$. 

\paragraph{Step 2.} Normalize $U$ so that $U(0, 0, \ldots)=U({\bf 0})=0$. Since $0$ is in the interior of $u(X)$, and since $U({\bf v} \lambda {\bf 0})=\lambda U({\bf v})$ for any ${\bf v}\in \mathbb{R}^\infty$ and for every $\lambda \in (0, 1)$, we can assume that $U$ is defined on $\mathbb{R}^{\infty}$.

Mixture linearity of $U$ implies standard linearity of $U$ on $\mathbb{R}^{\infty}$. To prove this, we need to show that $U(k {\bf v})= kU({\bf v})$ for any $k$ and $U({\bf v}+{\bf u})=U({\bf v})+U({\bf u})$ for any ${\bf u, v} \in \mathbb{R}^{\infty}$. 

As $u(X)^{\infty}$ is a mixture set under the operation of taking convex combinations, $U({\bf v}k{\bf 0})=U(k{\bf v}+(1-k){\bf 0})=U(k {\bf v})=k U({\bf v})$ for any $k \in (0, 1)$.
If $k>1$ then $U({\bf v})=U(\frac{k}{k}{\bf v})=\frac{1}{k}U(k{\bf v})$. Multiplying both parts of this equation by $k$, we obtain $U(k {\bf v})=k U({\bf v})$ for all $k>1$. Therefore, $U(k {\bf v})= kU({\bf v})$ for any $k>0$.

To prove that $U({\bf v}+{\bf u})=U({\bf v})+U({\bf u})$, consider the mixture ${\bf v}\frac{1}{2}{\bf u}$.
By mixture linearity of $U$ we have:
\begin{equation}
\label{eq:ml1}
U({\bf v}\frac{1}{2}{\bf u})=\frac{1}{2}U({\bf v})+\frac{1}{2}U({\bf u})=\frac{1}{2}\left(U({\bf v})+U({\bf u})\right).
\end{equation} 
On the other hand, ${\bf v}\frac{1}{2}{\bf u}=\frac{1}{2}{\bf v}+\frac{1}{2}{\bf u}=\frac{1}{2}({\bf v}+{\bf u})$. Therefore,
\begin{equation}
\label{eq:ml2}
U({\bf v}\frac{1}{2}{\bf u})=U\left(\frac{1}{2}({\bf v}+{\bf u})\right)= \frac{1}{2} U({\bf v}+{\bf u})
\end{equation}
Comparing \eqref{eq:ml1} and \eqref{eq:ml2} we conclude that $U({\bf v}+{\bf u})=U({\bf v})+U({\bf u})$.

Finally, note that 
\[
U({\bf 0})=U\left({\bf v}+({\bf -v})\right)=U({\bf v})+U({\bf -v})=0,
\] 
hence $U({\bf -v})=-U({\bf v})$. Therefore, if $k<0$, then $U(k {\bf v})= -kU({\bf -v})=kU({\bf v})$.

For each $T$, consider the function $f \colon \mathbb{R}^T \to \mathbb{R}$ defined as follows:
\[
f(v_1, \ldots, v_T)=U(v_1, \ldots, v_T, 0, 0, \ldots). 
\]
This function is linear on $\mathbb{R}^T$ and it satisfies $f({\bf 0})=0$, therefore,
\[ 
f(v_1, \ldots, v_T)=\displaystyle\sum_{t=1}^T w_t^T v_t,
\] 
where ${\bf w}^T=(w^T_1, \ldots, w^T_T)$. By monotonicity $w^T_t \geq 0$ for all $t \leq T$.

Note that $w^T_t=U([1]_t)$, where $[1]_t$ is the vector with 1 in period $t$ and $0$ elsewhere. It follows that $w^{T}_t=w^{T^{\prime}}_t$ for any $T$ and $T^{\prime}$.
Hence there is a vector ${\bf w} \in \mathbb{R}^\infty$ such that $U(v_1, \ldots, v_T, 0, 0, \ldots)=\displaystyle\sum_{t=1}^{\infty} w_t v_t$ for any $(v_1, \ldots, v_T)\in \mathbb{R}^T$.

Recalling that $v_t=u(x_t)$ we obtain 
\[
U(u(x_1), \ldots, u(x_T), 0, 0, \ldots)=\sum_{t=1}^T w_t u(x_t) \text{ for all } {\bf x} \in X^{*}.
\]
Therefore, for every ${\bf x, y} \in X^*$ we have ${\bf x} \succcurlyeq {\bf y}$ if and only if 
\[
\sum_{t=1}^T w_t u(x_t) \geq \sum_{t=1}^T w_t u(y_t).
\]
By slightly abusing the notation, re-define $U$ so that:
\[
U(x_1, \ldots, x_T, x_0, x_0, \ldots)=\sum_{t=1}^T w_t u(x_t) \text{ for all } {\bf x} \in X^{*}.
\]
Hence $U({\bf x})=\displaystyle\sum_{t=1}^{\infty} w_t u(x_t)$ represents preferences on $X^*$.

\paragraph{Step 3.} Next, we show that $U(x_1, x_2, \ldots)$ converges for any $(x_1, x_2, \ldots)$. 
Define $U_T\colon X^\infty \to {\mathbb R}$ as follows: $U_T({\bf x}) = \displaystyle \sum_{t=1}^T u_t(x_t)$, where $u_t(x_t)=w_t u(x_t)$. Consider the sequence of functions $U_1, U_2, \ldots, U_T, \ldots$
According to the \linebreak Cauchy Criterion, a sequence of functions $U_T({\bf x})$ defined on $X^\infty$ converges on $X^\infty$ if and only if for any $\varepsilon>0$ and any ${\bf x} \in X^\infty$ there exists $T \in \mathbb{N}$ such that $\lvert U_N({\bf x}) - U_M({\bf x}) \rvert < \varepsilon$ for any $N, M \geq T$.

Fix some ${\bf x}\in X^\infty$ and $\varepsilon>0$. Suppose that for some $k$ it is possible to choose $x^+, x^-$ such that $[x^+]_k \succ [x_k]_k \succ [x^-]_k$. By Step 2 the preference $[x^+]_k \succ [x_k]_k \succ [x^-]_k$ implies that $w_k>0$. Therefore, as $u$ is a continuous function, it is without loss of generality to assume that
\[
u_k(x^+)-u_k(x_k)<\varepsilon /2 \text{ and } u_k(x_k)-u_k(x^-)<\varepsilon /2.
\]
It follows that $u_k(x^+)-u_k(x^-)<\varepsilon, $ or $u_k(x^-)-u_k(x^+)>-\varepsilon.$
By Axiom I6 there exist $T^+$ and $T^-$ satisfying $k\leq \min \{T^-, T^+\}$ such that 
\[
{\bf x}_{k, N}^+ \succcurlyeq {\bf x} \succcurlyeq {\bf x}_{k, M}^-, \ \text{for all} \ N\geq T^+, \ M\geq T^-.
\]
Let $T^*=\max \{T^-, T^+\}$. It is necessary to demonstrate that $\lvert U_N({\bf x}) - U_M({\bf x}) \rvert < \varepsilon$ for any $N, M \geq T^*$. If $N=M$ the result is obviously true. If $N\neq M$ then it is without loss of generality to assume that $N>M$.
By the additive representation:
\[
U({\bf x}_{k,N}^+) \geq U({\bf x}_{k,M}^-).
\]
Expanding 
\[
u_k(x^+)+\sum_{t=1, t \neq k}^N u_t(x_t) \geq u_k(x^-)+\sum_{t=1,  t \neq k}^M u_t(x_t).
\]
By rearranging this inequality
\[
\sum_{t=M+1}^N u_t(x_t) \geq u_k(x^-)-u_k(x^+) > -\varepsilon.
\]
As $N>M \geq T^*$ it is also true that $U({\bf x}_{k, M}^+) \geq U({\bf x}_{k, N}^-$), hence
\[
\sum_{t=M+1}^N u_t(x_t) \leq u_k(x^+) - u_k(x^-) < \varepsilon.
\]
Note that
\[
\sum_{t=M+1}^N u_t(x_t) = U_N({\bf x})-U_M({\bf x}).
\]
Hence, $\lvert U_N({\bf x}) - U_M({\bf x}) \rvert < \varepsilon$ and it follows that $U({\bf x})$ converges by the Cauchy criterion.

Suppose now that it is not possible to find such $k$ that \linebreak $[x^+]_k \succ [x_k]_k \succ [x^-]_k $ for some $x^+, x^- \in X$. 
If $w_t=0$ for all $t$ then the result is trivial. Suppose that $w_t>0$ for some $t$. Then for every period $t$ for which $w_t>0$ we have 
\[
x_t \in X^e \equiv \{z \in X \ \colon \ z\succcurlyeq z' \ \text{for all} \ z'\in X \, \text{or} \ z'\succcurlyeq z \ \text{for all} \ z'\in X \}.
\]
For some $\lambda \in (0,1)$ replace $x_t$ with the mixture $x_t \lambda x_0$ for each $t$. Call the resulting stream ${\bf x}^*$. 
Then
\begin{equation*}
U_T({\bf x})-U_T({\bf x}^*)=\sum_{t=1}^T u_t(x_t)-\sum_{t=1}^T u_t(x_t \lambda x_0)=(1-\lambda)\sum_{t=1}^T u_t(x_t)=(1-\lambda)U_T({\bf x}).
\end{equation*}
By rearranging this equation it follows that $U_T({\bf x}^*)=\lambda U_T({\bf x})$. By the previous argument $U_T({\bf x}^*)$ converges, therefore, $U_T({\bf x})$ converges.

\paragraph{Step 4.} Show that $U({\bf x})$ represents the order on $X^\infty$. 
Suppose that ${\bf x} \succcurlyeq {\bf y}$, where ${\bf x}, {\bf y} \in X^\infty$. If for some $k, j$ it is possible to find $x^+, y^-$ such that  $[x^+]_k \succ [x_k]_k$ and $[y^-]_j \prec [y_j]_j$, then
$[x^+ \lambda x_k]_k \succ [x_k]_k$ for every $\lambda \in (0, 1)$ and $[y^- \mu y_j]_j \prec [y_j]_j$ for every $\mu \in (0,1)$. Let $x^*=x^+ \lambda x_k$ and $y^*=y^- \mu y_j$ for some $\lambda, \mu \in (0, 1)$. Denote
\[
{\bf x}_{k,N}^*=(x_1, \ldots, x_{k-1}, x^*, x_{k+1}, \ldots, x_N, x_0, x_0, \ldots),
\]
and
\[
{\bf y}_{j,M}^*=(y_1, \ldots, y_{j-1}, y^*, y_{j+1}, \ldots, y_M, x_0, x_0, \ldots).
\]
Then by Axiom I6, there exist $T^-, T^+$ such that 
\[
{\bf x}_{k,N}^*\succcurlyeq {\bf x}\succcurlyeq {\bf y} \succcurlyeq {\bf y}_{j,M}^*\
\]
for all $N\geq T^+$ and for all $M\geq T^{-}$.
Since ${\bf x}_{k,N}^* \succcurlyeq {\bf y}_{j,M}^*$ and $U$ represents $\succcurlyeq$ on $X^*$ we have:
\[
U({\bf x}_{k,N}^*)\geq U({\bf y}_{j,M}^*).
\]
By Step 3 we know that $U(x_1, \ldots, x_{k-1}, x^*, x_{k+1} \ldots)$ and \linebreak $U(y_1, \ldots, y_{j-1}, y^*, y_{j+1}, \ldots)$ converge, so
\[
U(x_1, \ldots, x_{k-1}, x^*, x_{k+1}, \ldots) \geq U(y_1, \ldots, y_{j-1}, y^*, y_{j+1}, \ldots).
\]
Recall that $x^*=x^+ \lambda x_k$ and $y^*=y^- \mu y_j$ for some $\lambda \in (0,1)$ and some $\mu \in (0,1)$. Since $\lambda$ and $\mu$ are arbitrary, it follows that $U({\bf x}) \geq U({\bf y})$.

If it is not possible to find $x^+$, $y^-$ such that $[x^+]_k \succ [x_k]_k$ and $[y^-]_j \prec [y_j]_j$, then either $w_t=0$ for all $t$, in which case $U({\bf x})=U({\bf y})$; or $x_t \succcurlyeq z'$ for all $z'\in X$ and all $t$ with $w_t>0$, in which case $U({\bf x}) \geq U({\bf y})$; or $z' \succcurlyeq y_t$ for all $z'\in X$ and all $t$ with $w_t>0$ in which case $U({\bf x}) \geq U({\bf y})$.

It is worth noting that as ${\bf x} \succcurlyeq {\bf y}$ implies $U({\bf x}) \geq U({\bf y})$, then, by Axiom I2 it follows that $w_t>0$ for at least one $t$. Therefore, $\sum_{t=1}^\infty w_t>0$. Normalizing by $1/\sum_{t=1}^\infty w_t$, we can assume that $\sum_{t=1}^\infty w_t=1$.

Next, assume that $U({\bf x}) \geq U({\bf y})$. 
Suppose that it is possible to find $k$ and $x^+, x^- \in X$ such that ${\bf x}_{k, N}^+\succcurlyeq {\bf x}\succcurlyeq {\bf x}_{k, N}^-$ for some fixed $N$.
By mixture continuity, the set $\{ \alpha \colon {\bf x}_{k,N}^+ \alpha {\bf x}_{k, N}^- \succcurlyeq \bf{x}\}$ is closed. By assumption ${\bf x}_{k, N}^+\succcurlyeq {\bf x}$, so it follows that $\alpha = 1$ is included into the set.
Analogously, the set $\{ \beta \colon {\bf x} \succcurlyeq {\bf x}_{k,N}^+ \beta {\bf x}_{k,N}^-\}$ is closed. In fact, $\beta = 0$ belongs to the set, as $ {\bf x} \succcurlyeq {\bf x}_{k, N}^-$.
Therefore, as both sets are closed, nonempty and form the unit interval, their intersection is nonempty. Hence, there exists $\lambda$ such that ${\bf x} \sim {\bf x}_{k, N}^+ \lambda {\bf x}_{k, N}^-$. 
Note that ${\bf x}_{k, N}^+ \lambda {\bf x}_{k,N}^- = (x_1, \ldots, x_{k-1}, x^+ \lambda x^-, x_{k+1} \ldots, x_N, x_0, x_0, \ldots)$. Let $x^+ \lambda x^- = x^*$.
Define ${\bf x}_{k, N}^*= (x_1, \ldots,x_{k-1}, x^*, x_{k+1}, \ldots, x_N, x_0, x_0, \ldots)$. Therefore, if there exist periods $k, j$ and outcomes $x^+, x^-, y^+, y^-\in X$ such that ${\bf x}_{k, N}^+ \succcurlyeq {\bf x} \succcurlyeq {\bf x}_{k, N}^-$ and ${\bf y}_{j, M}^+ \succcurlyeq {\bf y} \succcurlyeq {\bf y}_{j, M}^-$ for some $N$ and some $M$, we can find $\lambda, \mu \in [0, 1]$ such that ${\bf x} \sim {\bf x}_{k, N}^*$ and 
\[
{\bf y} \sim {\bf y}_{j, M}^* = (y_1, \ldots, y_{j-1}, y^*,y_{j+1}, \ldots, y_M, x_0, x_0, \ldots), 
\]
where $y^*=y^+\mu y^-$.
We have already shown that if ${\bf x}\succcurlyeq {\bf y}$ then $U({\bf x}) \geq U({\bf y})$. From ${\bf x} \sim {\bf x}_{k, N}^*$ and ${\bf y} \sim {\bf y}_{j, M}^*$ it therefore follows that: 
\[ 
U({\bf x}_{k, N}^*)= U({\bf x}) \text{ and } U({\bf y}_{j, M}^*)= U({\bf y}).
\]
Hence, from the assumption $U({\bf x}) \geq U({\bf y})$ we obtain:
\[
U({\bf x}_{k, N}^*) \geq U({\bf y}_{j, M}^*).
\]
Recall that $U$ is an order-preserving function on $X^*$.
Thus, ${\bf x}_{k, N}^* \succcurlyeq {\bf y}_{j, M}^*$. Since ${\bf x} \sim {\bf x}_{k, N}^*$ and ${\bf y} \sim {\bf y}_{j, M}^*$, we obtain ${\bf x} \succcurlyeq {\bf y}$.

Suppose now that there is no such $k, j$ or outcomes $x^+, x^-, y^+, y^-$ such that ${\bf x}_{k, N}^+\succcurlyeq {\bf x}\succcurlyeq {\bf x}_{k, N}^-$ and ${\bf y}_{j, M}^+\succcurlyeq {\bf y}\succcurlyeq {\bf y}_{j, M}^-$ for some $N$ and some $M$. Then, using Axiom I6, we can conclude that either $x_t \in X^e$ for every $t$ with $w_t>0$ or $y_t\in X^e$ for every $t$ with $w_t>0$.
Assume that there is only an upper bound to preferences; i.e., $X^e \equiv \{z \in X \ \colon \ z \succcurlyeq z' \ \text{for every} \ z'\in X\}$. Then $U({\bf x}) \geq U({\bf y})$ means that $x_t\in X^e$ whenever $w_t>0$. Therefore, $U({\bf x})=U({\bf \overline{x}})$, where ${\bf \overline{x}}=(\overline{x}, \overline{x}, \ldots)$ and $\overline{x} \in X^e$. Hence, it follows by monotonicity that ${\bf x}\succcurlyeq {\bf y}$. In the case when there is only a lower bound, i.e., $\underline{x} \in X^e \equiv \{z \in X \ \colon \ z' \succcurlyeq z \ \text{for every} \ z'\in X\}$, the argument is similar.

Next, suppose that $X$ is preference bounded above and below, i.e., there exist $\underline{x},\overline{x} \in X^e$ with $\overline{x}\succcurlyeq x \succcurlyeq \underline{x}$ for every $x\in X$. Assume that $U({\bf x}) \geq U({\bf y})$. We need to demonstrate that ${\bf x} \succcurlyeq {\bf y}$. By monotonicity and continuity there exist $\lambda, \mu \in [0, 1]$ such that ${\bf x} \sim {\bf \overline{x} }\lambda {\bf \underline{x}}$ and ${\bf y} \sim {\bf \overline{x}}\mu {\bf \underline{x}}$. Since by assumption $U({\bf x}) \geq U({\bf y})$ and $U$ represents the preference order on constant streams, we have $U({\bf \overline{x} }\lambda {\bf \underline{x}}) \geq U({\bf \overline{x} }\mu {\bf \underline{x}})$. By rearranging this inequality $(\lambda-\mu)(U({\bf \overline{x}})-U({\bf \underline{x}}))$, and using $U({\bf \overline{x}}) > U({\bf \underline{x}})$ it follows that $\lambda \geq \mu$. Therefore, as ${\bf x} \sim {\bf \overline{x}}\lambda {\bf \underline{x}}$ and ${\bf y} \sim {\bf \overline{x} }\mu {\bf \underline{x}}$ and $\lambda \geq \mu$, we conclude that ${\bf x} \succcurlyeq {\bf y}$.

Thus $({\bf w}, u)$ is an AA representation for $\succcurlyeq$.

\paragraph{Step 5.} Uniqueness of $w_t$. Assume that $({\bf w}^\prime, u^\prime)$ is another AA representation. Then, for any $t$ we have $w_t>0$ if and only if $w^\prime_t>0$.
 Consider the set of all constant programs $\{{\bf x}\in X^\infty: {\bf x} = (a, a, \ldots), \text{ where } a\in X\}$, which is a mixture set. Applying $({\bf w}^\prime, u^\prime)$ and $({\bf w}, u)$ to this set we conclude that $u(a)>u(b)$ if and only if $u^\prime(a)>u^\prime(b)$ for every $a, b \in X$. By Theorem 1  \cite{fishburn1982foundations} it implies that $u=A u^\prime+B$ for some $A>0$ and some $B$. 
Hence,
\[
 \sum_{t=1}^ \infty w_t u(x_t) \geq \sum_{t=1}^ \infty w_t u(y_t) \ \text{if and only if} \ \sum_{t=1}^ \infty w^\prime _t u(x_t) \geq \sum_{t=1}^ \infty w^\prime_t u(y_t).
\]
For any $t, s$ with $t \neq s$ and any $x', x'' \in X$, let $[x',x'']_{t,s}$ denote the stream with $x'$ in the t\textsuperscript{th} position, $x''$ in the s\textsuperscript{th} position and $x_0$ elsewhere. Fix $t, s$ with $w_t>0$ and $w_s>0$. Using non-triviality, choose some $x^+, x^- \in X$ such that $x^+ \succ x^-$. 
Define ${\bf x}=[x^+, x^+]_{t,s}, \ {\bf y}= [x^+, x^-]_{t,s}, \ {\bf z}= [x^-, x^-]_{t,s}$.
From the AA representation it follows that ${\bf x} \succ {\bf y} \succ {\bf z}$.
By continuity of the AA representation there exists $\lambda \in (0, 1)$ such that ${\bf y} \sim {\bf x} \lambda {\bf z}$.
Applying the AA representation to ${\bf y} \sim {\bf x} \lambda {\bf z}$ we obtain
\[
w_t u(x^+)+w_s u(x^-) =\lambda (w_t+w_s) u(x^+)+(1-\lambda) (w_t+w_s) u(x^-).
\]
It follows that $(1-\lambda)w_t = \lambda w_s$. Similarly, $(1-\lambda)w^\prime_t = \lambda w^\prime_s$. Therefore, $w_t/w_s =w^\prime_t/w^\prime_s$. As this is true for any $t, s$, we obtain that ${\bf w}=C {\bf w^\prime}$ for some $C>0$.
\end{proof}

{ \small
\section*{Acknowledgments}
I would like to thank my co-supervisor Matthew Ryan for thoughtful advice and support. I am grateful to Arkadii Slinko and Simon Grant for helpful discussions. I also would like to thank Australian National University for hospitality while working on the paper. Helpful comments from seminar participants at Australian National University, The University of Auckland and participants at Auckland University of Technology Mathematical Sciences Symposium (2014), Centre for Mathematical Social Sciences Summer Workshop (2014), the joint conferences "Logic, Game theory, and Social Choice 8" and "The 8th Pan-Pacific Conference on Game Theory" (2015) Academia Sinica are also gratefully acknowledged. Financial support from the University of Auckland is gratefully acknowledged.
}

\end{document}